\documentclass[prx,aps,twocolumn,notitlepage,superscriptaddress,showpacs,nofootinbib]{revtex4-2}

\usepackage{enumerate,appendix}
\usepackage{qcircuit}
\usepackage{amsmath, amsthm, amssymb}
\usepackage{color,calc,graphicx}
\usepackage[usenames,dvipsnames,svgnames,table,cmyk,hyperref]{xcolor}
\usepackage[colorlinks]{hyperref}
\usepackage{optidef}
\hypersetup{
	colorlinks = true,
	urlcolor = {blue},
	citecolor = {blue},
	linkcolor= {blue}
}

\usepackage{latexsym}

\usepackage{bbm}

\usepackage[charter,cal=cmcal,sfscaled=false]{mathdesign}
\usepackage{booktabs}
\graphicspath{ {./images/} }
\usepackage{multirow} 
\usepackage{dcolumn}
\usepackage{mathrsfs}
\usepackage{csvsimple-l3}

\def \be {\begin{equation}}
\def \ee {\end{equation}}

\newcommand{\ket}[1]{|#1\rangle}

\def\>{\rangle}
\def\<{\langle}

\newtheorem{definition}{Definition}
\newtheorem{theorem}{Theorem}
\newtheorem{lemma}[theorem]{Lemma}

\usepackage{algpseudocode}

\usepackage[most]{tcolorbox}

\usepackage{makecell}
\usepackage{pifont}

\begin{document}

\title{Learning unknown stabilizer codes using product measurements}

\author{Heather Leitch}
\email{h.leitch@sheffield.ac.uk}
\author{Sri.S.Tirukkovalluri}
\author{Yingkai Ouyang}
\email{y.ouyang@sheffield.ac.uk}
\affiliation{School of Mathematical and Physical Sciences, University of Sheffield, Sheffield, S3 7RH, United Kingdom}

\begin{abstract} 
Efficiently characterizing quantum error correcting codes is a key challenge on the path to fault-tolerant quantum computation. Stabilizer codes, a central class of such codes, are defined by a set of stabilizer generators. Here, we present an algorithm that uses random single-qubit measurements to learn the stabilizer generators of any stabilizer code from $N$ copies of stabilizer states in its codespace, requiring no prior knowledge of the code's structure. This also enables verification that a device implements its intended code. We derive a lower bound on $N$ needed to recover the stabilizer generators with high probability, together with a bound on the algorithm's overall probability of success. When applied to quantum low-density parity-check (qLDPC) codes, a leading candidate for practical fault-tolerant architectures, our approach has query complexity that scales polynomially with $n$, the number of qubits.
\end{abstract}

\maketitle

\section{Introduction}

Stabilizer codes are one of the most widely studied classes of quantum error correction codes. The stabilizer formalism, first introduced by Gottesman \cite{gottesman1997stabilizer}, defines the codespace as the simultaneous $+1$ eigenspace of a set of mutually commuting Pauli operators, called stabilizers. The Calderbank-Shor-Steane (CSS) codes construction \cite{calderbank1996good,steane1996simple} then showed how quantum stabilizer codes could be built from pairs of classical linear codes, leading to the well known $[[7,1,3]]$ Steane code \cite{steane1996multiple} and $[[9,1,3]]$ Shor code \cite{shor1995scheme}. Arguably, one of the biggest advancements in stabilizer codes is the surface code \cite{bravyi1998quantum}, which to this day is one of the leading architectures for experimental implementation for fault tolerant quantum computing \cite{krinner2022realizing,google2023suppressing,google2025quantum}, as it is based on a 2D lattice and involves only nearest neighbour interactions.
Furthermore, qLDPC codes have near-capacity performance \cite{kuo2026degenerate} under the independent and identically distributed quantum erasure model \cite{grassl1997codes,wu2022erasure,gu2025fault}.

A key drawback of surface codes is that they require a large number of physical qubits per logical qubit. Quantum low-density parity-check (qLDPC) codes are a class of stabilizer codes in which every stabilizer generator has bounded weight and each physical qubit participates in a bounded number of stabilizer generators \cite{breuckmann2021quantum}. It has been shown that asymptotically good qLDPC codes exist \cite{panteleev2022asymptotically}, meaning that both the number of logical qubits and the code distance scale linearly with the number of physical qubits ($n$).

Quantum state tomography (QST) aims to obtain a complete description of an unknown quantum state by preparing multiple identical copies of the state and performing measurements in different bases. Learning quantum states is a central tool for the verification and benchmarking of quantum devices \cite{paris2004quantum,cramer2010efficient,ivanova2023optimal,blume2025quantum}.
A fundamental limitation of QST is that the number of measurements required to fully characterize an arbitrary quantum state scales exponentially with the number of qubits \cite{cramer2010efficient, christandl2012reliable}. 
To overcome this limitation, quantum state learning has been formulated within the framework of Probably Approximately Correct (PAC) learning, where the goal is to output a state that is within an error $\epsilon$ of the true state with probability at least $1-\delta$. This allows approximate state learning with substantially fewer resources than required for full state tomography \cite{aaronson2007learnability,lowe2022lower}.

Recent advances at the intersection of computational learning theory and quantum information science have explored the feasibility of learning various quantum states \cite{anshu2024survey, lai2022learning, arunachalam2022optimal}, along with different measurement approaches \cite{aaronson2018online, lowe2021learning}. In the case of stabilizer states, improved guarantees are known: sample complexity scales as $O(n)$ with collective measurements (joint measurements on multiple copies of the state that are information-efficient but experimentally demanding) and as $O(n^2)$ with single-copy measurements, where each copy is measured separately \cite{pirsa_PIRSA:08080052}. These bounds were recently extended to a bounded quantum memory setting \cite{arunachalam2026optimal}. Alternative approaches achieve $O(n)$ copies with polynomial runtime using Bell measurements \cite{montanaro2017learning}.

Stabilizer states are also efficiently PAC-learnable \cite{rocchetto2017stabiliser}, meaning the run-time of their learning algorithms is polynomial in $n$, the number of qubits.
Recent research has explored the learning of graph states, which are quantum states represented by graphs and under local Clifford transformations are equivalent to stabilizer states  \cite{schlingemann2001stabilizer, van2005local, zeng2007local}.  Graph state learning has been studied using collective measurements \cite{montanaro2020quantum}, as well as more restricted measurement models. In particular, a product-measurement-based learning algorithm was recently introduced \cite{ouyang2022learning}, where measurements are restricted to independent local measurements on individual qubits, making them experimentally simpler but generally less information is learned compared to  global measurements.

Recent work has studied the closely related problem of verifying fidelity to a known QEC code subspace using local measurements \cite{chen2025quantum}, but this assumes the code is already known. 
In this paper, we move beyond learning individual stabilizer states and focus on learning the stabilizer generators that define stabilizer codes. With recent advances in experimental quantum computing and the first implementations of quantum error correction, reliable methods for verifying and characterizing quantum codes have become increasingly important. 
We address this problem by proposing an algorithm that can take any $N$ states from a codespace and efficiently learn the codes corresponding stabilizer generators, thereby providing an approach to code verification.
This algorithm works for any stabilizer code, but we show that it is particularly effective for qLDPC codes. We also establish a lower bound on $N$, the number of states in the codespace, needed to reliably recover all the stabilizer generators.
The algorithm we introduce for learning stabilizer codes uses product measurements, meaning we take a single-qubit measurement of each qubit.
Specifically, we focus on learning $n$ independent generators of an $n$-qubit stabilizer codeword by performing randomized product measurements on multiple copies of the state. Our algorithm uses random sampling approach, with qubits of the state are measured in repeated measurement patterns randomly chosen from different bases ($X$, $Y$, $Z$). By analysing the parity of these  measurement outcomes across various qubit subsets of the quantum state (codeword), we identify the $n$ stabilizer generators that define the codeword.

We begin by defining a stabilizer code in terms of its stabilizer generators and codewords in section \ref{sec:results}, focusing specifically on qLDPC codes. Building on these definitions, we develop algorithms for learning the stabilizer generators for any stabilizer code, from any $N$ stabilizer states in the codespace, presented with illustrative examples in sections \ref{sec:alg1sum} and \ref{sec:alg2sum} (with explicit algorithms in section \ref{sec:method}). An overview of this approach is presented in Fig. \ref{fig:Diagram}. To understand when these algorithms successfully recover the stabilizer generators, we state Theorem \ref{th:Nbound}  (with proof in section \ref{sec:Nbound}) which gives a lower bound on $N$, the number of codewords needed. This is followed by an upper bound on the overall probability that our algorithms fail, stated in Theorem \ref{th:pfailure} (with proof in section \ref{sec:proof6}), which leads to the conclusion that for qLDPC stabilizer codes, only a polylogarithmic number of copies of the stabilizer states are required to learn the stabilizer generators.

This protocol therefore gives a practical method for verifying that the qLDPC codes implemented on near-term devices are indeed the codes intended, without requiring an impractically large number of copies of codewords.
\begin{figure*}[htbp]
\centering
\makebox[0pt]{\includegraphics[width=0.9\paperwidth]{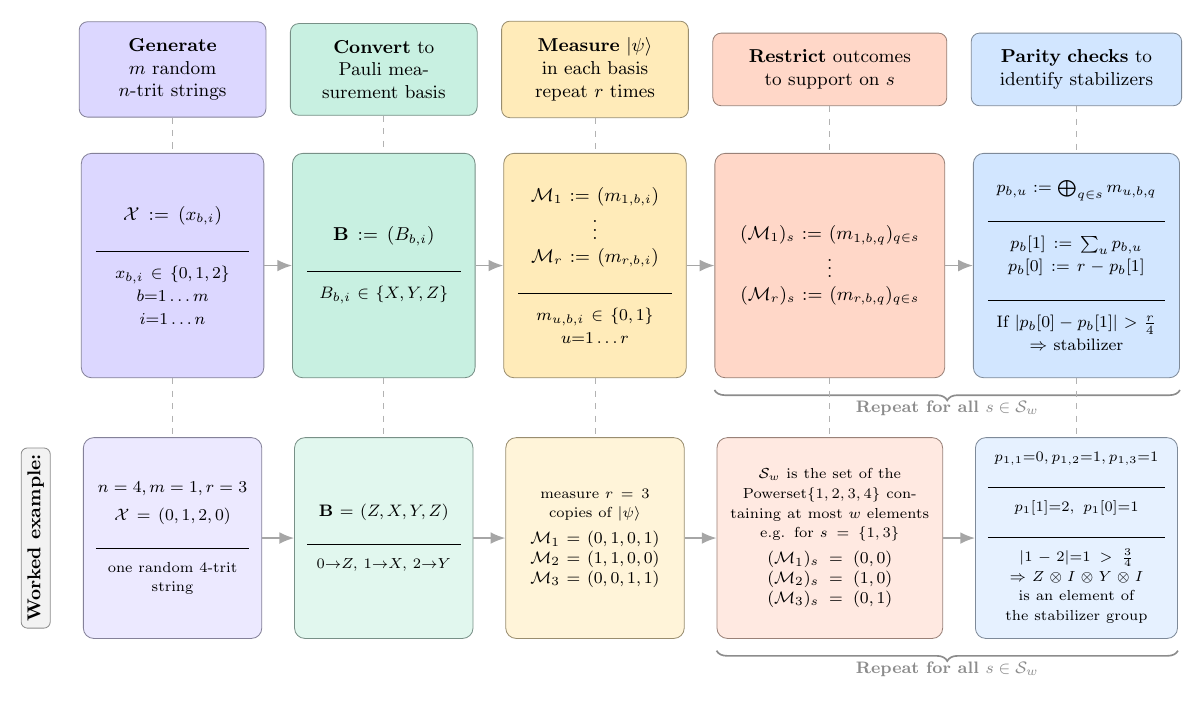}}
\caption{Overview of the algorithm which allow us to identify the stabilizer generators of a stabilizer code, given any $N$ codewords. The first $3$ steps are explained in further detail in sec. \ref{sec:alg1sum} and the final $2$ in sec. \ref{sec:alg2sum}.}
\label{fig:Diagram}
\end{figure*}

\section{Main Results}\label{sec:results}
Here we present an algorithm, given in two parts and summarized in Fig. \ref{fig:Diagram}, which takes as input any $N$ stabilizer states in the codespace of a stabilizer code and outputs the stabilizer generators defining the code. We derive a bound on the number of copies of the stabilizer state required for the algorithm to succeed, as well as the effect of depolarizing noise and therefore give a bound on the algorithm's overall probability of success. We begin, however, by defining stabilizer codes and their codewords.
\begin{definition}\label{def:stabilizer}
A stabilizer code can be defined by parameters $[[n,k,d]]$, where $n$ is the number of physical qubits, $k$ is the number of logical qubits and $d$ is the distance.
The code is defined by a set of independent commuting Pauli operators ${G_1,\dots,G_g}$, called the generators of the code. Here, $g=n-k$ is the number of independent generators.
\end{definition}
\begin{definition}\label{def:codeword}
The codewords are the states $\ket{\psi}$ satisfying
\begin{equation}
G_j\ket{\psi} = \ket{\psi},\;\;\;\;\; \forall j = 1,\dots,g.
\end{equation}
\end{definition}

\begin{tcolorbox}[enhanced jigsaw,breakable,pad at break*=1mm,
  colback=blue!5!white,colframe=blue!75!black,title=Example I: Stabilizer generators and codewords]
  The $[[4,2,2]]$ quantum error correcting code \cite{rains2002quantum,gottesman1997stabilizer} has two stabilizer generators
\begin{equation}\label{eq:gen}
\langle XXXX,ZZZZ\rangle
\end{equation}
and four logical codewords 
\begin{align}
    |0_L\> &= (|0000\> + |1111\>)/\sqrt 2\\
    |1_L\> &= (|0011\> + |1100\>)/\sqrt 2\\
    |2_L\> &= (|0110\> + |1001\>)/\sqrt 2\\
    |3_L\> &= (|1010\> + |0101\>)/\sqrt 2,
\end{align}
with the codespace being defined by the span of these four logical codewords.
\end{tcolorbox}

We now focus our attention on a subset of stabilizer codes, quantum low density parity check codes.
\begin{definition}\label{def:qldpc}
Quantum low-density parity check (qLDPC) codes \cite{breuckmann2021quantum,pecorari2025high} are stabilizer codes for which there exist stabilizer generators such that both the number of qubits acted on nontrivially by each stabilizer generator and the number of stabilizer generators acting on each qubit are bounded by a constant. 
\end{definition}

\subsection{Algorithm 1 Summary}\label{sec:alg1sum}
Generate $m$ random $n$-trit  strings, ${\bf x}=(x_1,\dots,x_n)$, where $x_i\in \{0,1,2\}, \forall i=1,\dots,n$.
These strings define an $m\times n$ matrix $\mathcal{X}$ whose $b$th row, denoted $(\mathcal{X})_b$, specifies an $n$-trit string, with $b=1,\dots,m$. 
Each $n$-trit string corresponds to a measurement basis ${\bf B}=\bigotimes_{i=1}^n B_i$, where
\begin{equation}
B_i = 
\begin{cases}
Z & \text{if}\; \; x_i = 0\\
X & \text{if} \;\; x_i = 1\\
Y & \text{if} \;\; x_i = 2
\end{cases}.
\end{equation}

For each measurement basis ${\bf B}$, we take $r$ copies of the state $\ket{\psi}$ and measure qubit $i$ in the basis $B_i$. 
Each repetition produces an $n$-bit measurement outcome string ${\bf m}=(m_1,\dots,m_n)$, where $m_i$ is equal to 0 if the eigenvalue of the corresponding $B_i$ is $+1$, and $1$ if the eigenvalue is $-1$.
For each measurement basis $(\mathcal{X})_b$, the corresponding outcome string forms the $b$th row of an $m\times n$ measurement matrix $\mathcal{M}$. 
Repeating the measurements $r$ times yields a collection of matrices $\mathcal{M}_u$, where $u=1,\dots,r$.

\begin{tcolorbox}[enhanced jigsaw,breakable,pad at break*=1mm,
  colback=blue!5!white,colframe=blue!75!black,title=Example II: Algorithm $1$]
  Generate $m$ random 4-trit strings
  \begin{equation}\label{eq:alg1}
m\begin{cases}\begin{pmatrix}
    0120\\
    \vdots\\
    1111\\
    \vdots\\
    0000\\
    \vdots
\end{pmatrix}
\end{cases}
\xrightarrow{\substack{\text{Corresponding to}\\ \text{measurement in} \\ \text{ the basis}}}
\begin{pmatrix}
    ZXYZ\\
    \vdots\\
    XXXX\\
    \vdots\\
    ZZZZ\\
    \vdots
\end{pmatrix}.
\end{equation}
For the sake of this example, lets assume we have $N$ copies of the stabilizer state $\ket{\psi}\equiv \ket{0_L}$. For each $4$-trit string, we take $r$ copies of the state $\ket{\psi}$ and measure each
qubit in the corresponding basis:
\begin{center}
\setlength{\tabcolsep}{3pt}
\small
\begin{tabular}{|c|c|c|}
\hline
\makecell{Measurement \\ Basis} & \makecell{Possible Measurement \\ Outcomes} & \makecell{Outcome \\ Probabilities} \\ 
\hline
$ZXYZ$ & 
\makecell{$0000,\;0010,\;0100,\;0110$ \\ $1001,\;1011,\;1101,\;1111$} & $1/8$ \\ 
\hline
$XXXX$ & \makecell{$0000,\;0011,\;0101,\;0110$ \\ $1001,\;1010,\;1100,\;1111$} & $1/8$ \\ 
\hline
$ZZZZ$ & $0000, \; 1111$  & $1/2$ \\ 
\hline
\end{tabular}
\end{center}
\end{tcolorbox}

\subsection{Algorithm 2 Summary}\label{sec:alg2sum}
Let us define 
\begin{equation}
\mathcal{S}_w:=\left\{ S \subseteq \{1,\dots,n\} : |S| \le w \right\},
\end{equation}
which is the powerset of $\{1,\dots,n\}$ containing sets of at most $w$ elements.
Each measurement outcome is an $n$-bit string. If ${\bf m}=(m_1,\dots,m_n)$
denotes such an $n$-bit string, 
we let ${\bf m}_{s} = (m_j)_{j \in s}$ denote a substring of ${\bf m}$ with components from the set $s \in \mathcal{S}_w$.
Given any bit-string, we define the parity of the bit string as its Hamming weight modulo 2. If the parity of a bit-string is 0, we say that it has even parity, and increase $p[0]$ by $1$. If the parity of a bit-string is 1, we say that it has odd parity and increase $p[1]$ by $1$.
For every $s \in \mathcal{S}_w$
\begin{equation}
    {\bf B}_{s} = \bigotimes_{i \in s}B_i
\end{equation}
will have $r$ associated ${\bf m}_{s}$ measurement substrings.
We measure the parity of each substring and count how many have even and odd parity. 
In the case of no measurement errors, we claim that if ${\bf B}_{s}$ is an element of the stabilizer set, then all measurement outcomes supported on $s$ have even parity, that is, $p_K[0] = r$. However, as we want to allow some probability $q$ of a measurement error, we can only assume $\mathbb{E}(p[0]) = r$. Similarly, if $-{\bf B}_{s}$ belongs to the stabilizer, then $\mathbb{E}(p_K[1]) = r$.
We therefore assume that if $|p[0]-p[1]|>r/4$, the corresponding $\pm {\bf B}_{s}$ is an element of the stabilizer set.
\begin{tcolorbox}[enhanced jigsaw,breakable,pad at break*=1mm,
  colback=blue!5!white,colframe=blue!75!black,title=Example III: Algorithm $2$]
The powerset of $\{1,\dots,4\}$ containing up to $w=4$ elements is
\begin{equation}\label{eq:powersets}
\begin{split}
\mathcal{S}_4 =&\{ \{1\},\{2\},\{3\},\{4\},\{1,2\},\{1,3\},\{1,4\},\\
&\{2,3\},\{2,4\},\{3,4\},\{1,2,3\},\{1,2,4\},\\
&\{1,3,4\},\{2,3,4\},\{1,2,3,4\} \}.
\end{split}
\end{equation}
Algorithm $2$ will check, for every $s \in \mathcal{S}_4$, ${\bf B}_{s}$ and count the parity of all $r$ associated ${\bf m}_{s}$ measurement substrings. 
The table below shows a few examples of measurement substrings of  the measurement bases $ZXYZ$, $XXXX$ and $ZZZZ$ and the values of the parity that we would expect
\begin{center}
\setlength{\tabcolsep}{1pt}
\small
\begin{tabular}{|c|c|c|c|c|c|}
\hline
\makecell{Measurement \\ Basis} & Substring & $\mathbb{E}(p[0])$ & $\mathbb{E}(p[1])$ & $\mathbb{E}(\Delta p)$ & Stabilizer\\ 
\hline
\multirow{3}{*}{$ZXYZ$} & $\{1,2\}$ &$r/2$ & $r/2$ & $0$ &  \ding{55}\\ 
\cline{2-6}
  & $\{1,4\}$ &$r$ & $0$ & $r$ &  \ding{51}\\ 
\cline{2-6}
  & $\{1,2,3,4\}$ &$r/2$ & $r/2$ & $0$ &  \ding{55}\\ 
\hline
\multirow{2}{*}{$XXXX$} & $\{2,3\}$ & $r/2$ & $r/2$ & $0$ & \ding{55}  \\ 
\cline{2-6}
  & $\{1,2,3,4\}$ & $r$ & $0$ & $r$ & \ding{51}  \\ 
\hline
\multirow{3}{*}{$ZZZZ$} & $\{1\}$ & $r/2$ & $r/2$ & $0$ & \ding{55} \\ 
\cline{2-6}
  & $\{3,4\}$ & $r$ & $0$ & $r$ & \ding{51} \\ 
\cline{2-6}
 & $\{1,2,3,4\}$ & $r$ & $0$ & $r$ & \ding{51} \\ 
\hline
\end{tabular}
\end{center}
where we have defined $\Delta p = |p[0]-p[1]|$. 
We therefore conclude that $ZIIZ$, $XXXX$, $IIZZ$ and $ZZZZ$ are highly likely to be elements of the stabilizer set.
\\
\noindent
\underline{Note}: if we only use $N$ copies of the state $\ket{0_L}$ then there are $2^4$ possible elements of the stabilizer set we may find, not just $XXXX$ and $ZZZZ$.
\end{tcolorbox}

\subsection{Bound on the Number of Stabilizer States}\label{sec:Nbound}
For any qLDPC code, we can find a a set of stabilizer generators which have weight at most $w$. As algorithms $1$ and $2$ sample possible Pauli operators, we can limit our search to Pauli operators of weight at most $w$. 
\begin{definition}\label{def:A}
There are $A$ possible non-trivial Pauli operators of weight at most $w$ for an $n$-qubit system
\begin{equation}
A = \sum_{k=1}^w \binom{n}{k}3^k.
\end{equation}
\end{definition}
We then find that, for qLDPC codes, we can derive a bound on the number of codewords, $N$, we will need to reliably learn the stabilizer generators.
\begin{theorem}\label{th:Nbound}
Consider a stabilizer code with a particular set of $g$ (not necessarily unique) generators $G_1, \dots, G_g$ that we wish to learn. All of these generators will have weight at most $w$. Using algorithms 1 and 2, we require 
\begin{equation}
    N \ge 2A\ln(g)
\end{equation}
copies of the stabilizer state $\ket{\psi}$ to ensure we learn all $g$ stabilizer generators with probability at least $1/g$.
\end{theorem}

\subsection{Bound on the Probability of Success}\label{sec:successbound}
No physical computation is free from error, so we assume the qubit is subject to noise prior to measurement. We model this as a depolarizing channel: the identity operator is applied with probability $(1-3p/2)$, and each of a bit-flip ($X$), phase-flip ($Z$), or combined bit- and phase-flip ($Y$) error is applied with probability $p/2$. The qubit is then measured in the $X,Y$ or $Z$ basis with equal probability. Each error type is undetectable in its own basis but flips the outcome in the other two: an $X$ error flips the result of a $Z$ or $Y$ measurement, a $Z$ error flips the result of an $X$ or $Y$ measurement, and a $Y$ error flips the result of an $X$ or $Z$ measurement. Consequently, for any fixed measurement basis, exactly two of the three error types (each occurring with probability $p/2$) cause a flip, giving a total flip probability of $p/2+p/2=p$. Since this holds regardless of which basis is chosen, a bit-flip error on the outcome of a single measurement occurs with probability $p$ overall.
\begin{lemma}\label{lem:q}
If bit-flip errors on a single measurement outcome occur with probability $p$, then the outcome of any product measurement of a weight $w$ Pauli operator has probability $q$ of exhibiting a parity error, i.e. measuring a $0$ rather than $1$ or vice versa.
\begin{equation}\label{eq:q}
q= \frac{1 - (1-2p)^w}{2}
\end{equation}
\end{lemma}
For proof of lemma \ref{lem:q} see section \ref{sec:proof2}.
We measure $r$ copies of each Pauli operator, $Q$, and check if it has even or odd parity. 
Here, the parity of the product measurement is defined to be even if the number of +1 eigenvalues is obtained an even number of times, and defined to be odd if the number of -1 eigenvalues is obtained an odd number of times. We define the number of measurements of even parity as $P[0]$ and odd parity as $P[1]$. If there is a significant difference in parity counts, i.e. if $\left| P[0] -P[1] \right| >r/4$ then we assume $Q$ must be a member of the stabilizer group $G$, otherwise we reject $Q$.
\begin{lemma}\label{lem:paccept}
If the Pauli operator $Q$ is an element of the stabilizer group $G$, then the probability that it will correctly be accepted by our algorithm is
\begin{equation}
\Pr\left( \text{accept} | Q \in G \right) \geq 1-2 c_q^{5r/8},
\end{equation}
where 
\begin{equation}
c_q = \frac{8(1-q)}{5}e^{1-\tfrac{8(1-q)}{5}}.
\end{equation}
\end{lemma}

\begin{lemma}\label{lem:preject}
If the Pauli operator $Q$ is not an element of the stabilizer group $G$, then the probability that our algorithm correctly rejects it is 
\begin{equation}
\Pr\left( \text{reject} | Q\notin G\right) \geq 1-2\left( \frac{64}{27 e}  \right)^{r/8}.
\end{equation}
\end{lemma}

Proofs of lemmas \ref{lem:paccept} and \ref{lem:preject} can be found in sections \ref{sec:proof3} and \ref{sec:proof4}.

\begin{theorem}\label{th:psuccess}
Following from lemmas \ref{lem:paccept} and \ref{lem:preject}, the probability that the algorithm is successful, i.e. correctly identifies a Pauli operator $Q$ as an element of the stabilizer group $G$, and rejects it if not an element of $G$ is
\begin{equation}
\Pr(\text{success}) \geq 1-2\max \left\{ c_q^{5r/8}, \left(\frac{64}{27 e}  \right)^{r/8} \right\}.
\end{equation}
\end{theorem}
Notice that, as long as we are able to achieve a parity error $q<0.217087$, and therefore a bit-flip error $p<\tfrac{1}{2}\left( 1-0.565825^{1/w} \right)$, then 
\begin{equation}\label{eq:qrequirement}
    \max \left\{ c_q^{5r/8}, \left(\frac{64}{27 e}  \right)^{r/8} \right\} = \left(\frac{64}{27 e}  \right)^{r/8}. 
\end{equation}
The bivariate bicycle (BB) codes \cite{bravyi2024high} are a good candidate for verification using our algorithm, as they form a family of weight-6 stabilizer codes. As stated in \cite{bravyi2024high}, a physical error rate of $10^{-3}$ is a realistic target for near-term demonstrations of BB codes. For $w=6$, Eq. \eqref{eq:qrequirement} holds for $p < 0.0452733$. This corresponds to a bound on the depolarising noise $3p/2<0.0679099$.

For each measurement pattern, we include up to weight $w$ Pauli operators, and therefore the probability that there is at least one incorrect acceptance of a stabilizer element is $2 A \left(\frac{64}{27 e}  \right)^{r/8}$, therefore after $N$ rounds, we find
\begin{equation}
\Pr(\text{failure}) \leq 2 NA\left(\frac{64}{27 e}  \right)^{r/8}.
\end{equation}
Theorem \ref{th:Nbound} tells us that we want $N \ge 2A\ln(g)$.
We can set $N = 2A\ln(n)$.
Then the upper bound of the failure probability is at most 
\begin{equation}
 \Pr(\text{failure}) \leq 4 A^2 \ln(n) \left(\frac{64}{27 e}  \right)^{r/8} .\label{pfail}
\end{equation}

\begin{theorem}\label{th:pfailure}
For our protocol, choosing
\begin{equation}
r \geq 76 \left( 2w\ln(3n^2) + \ln(4w^2)+\ln\left(\ln(n)\right)\right)   
\end{equation}
gives an upper bound on the failure probability
\begin{equation}
\Pr(\text{failure})\leq n^{-2w}
\end{equation}
\end{theorem}
A proof of theorem \ref{th:pfailure} can be found in \ref{sec:proof6}.
Therefore, whenever $r=\Theta\left( w \ln(n) \right)$, the failure probability $\Pr(\text{failure}) \rightarrow 0$ as $n\rightarrow 0$.

The memory requirements of the algorithm are linear in the code length. At any stage, the algorithm stores only the learned stabilizer information, which contains $O(n)$ entries. The dominant computational cost arises from enumerating the set of candidate stabilizer generators up to weight $w$, whose size is bounded by $A=O(n^w)$. Since $w$ is constant for qLDPC codes, this contribution is polynomial in $n$. Combining this with the query complexity $N=\Theta(w\ln(n))$, the overall computational complexity is polynomial in $n$. Therefore both the computational complexity and memory requirements remain efficient for qLDPC codes.

There are several existing protocols for quantum state verification (QSV) of stabilizer states, such as \cite{dangniam2020optimal}, which verifies that a given state has infidelity at most $\epsilon$ with respect to a known target stabilizer state, with sample complexity $N=O\left(\tfrac{3}{2\epsilon}\ln\left(\tfrac{1}{\delta}\right)\right)$, where $\delta$ is the failure probability.
Similarly, \cite{takeuchi2018verification} introduces an adaptive protocol that verifies whether a multi-qubit state is within some tolerance of a target state. This is done via sequential single-qubit Pauli measurements.  
Both approaches, however, assume that the target stabilizer generators are known in advance, and can only test proximity to a fixed state rather than learn any of its stabilizer generators.
More closely related to our setting is subspace-verification-assisted fidelity estimation (SVAFE) \cite{chen2025quantum}, which extends QSV from individual states to code subspaces. 
Given a target code, the SVAFE protocol has sample complexity scaling as $O\left(\tfrac{g}{\epsilon}\ln\left(\tfrac{1}{\delta}\right)\right)$ for generic stabilizer codes with $g$ stabilizer generators.
Here again, however, both the code subspace and its stabilizer generators are assumed to be known in advance.
Our algorithm addresses the converse problem, which to our knowledge has not previously been studied: given copies of a stabilizer state from an unknown code space, it recovers a valid set of stabilizer generators using $N=\Theta(w\ln(n))$ copies. This demonstrates that verification of a stabilizer QEC code can be achieved with only polylogarithmically many copies of the stabilizer state, even without prior knowledge of its generators.

\section{Technical Details: Algorithms and Proofs}\label{sec:method}
In this section, we give a formal treatment of the results introduced in Section \ref{sec:results}. We begin by presenting explicit descriptions of the two algorithms outlined in Sections \ref{sec:alg1sum} and \ref{sec:alg2sum}. We then provide proofs of Theorem \ref{th:Nbound} and \ref{th:pfailure} and Lemmas \ref{lem:q}, \ref{lem:paccept} and \ref{lem:preject}, establishing the bounds on the number of stabilizer states required and the probability of success stated in sections \ref{sec:Nbound} and \ref{sec:successbound} respectively.

\begin{figure}[h]
\centering
\begin{minipage}{0.95\columnwidth}

\hrule height 0.5pt
\vspace{4pt}

\textbf{Algorithm 1: Measurement Process for Stabilizer State Identification}

\vspace{4pt}
\hrule height 0.5pt
\vspace{2pt}

\begin{algorithmic}[1]
\For{$b = 1$ to $m$} \Comment{Iterate over measurement patterns}
    \State Generate $\mathbf{x}_b = (x_{b,1}, \ldots, x_{b,n})$ \Comment{Random $n$-trit string}
    \State $x_{b,i} \in \{0, 1, 2\}$ for $i = 1, \ldots, n$
    \Comment{0: Z, 1: X, 2: Y basis}
    \For{$u = 1$ to $r$}
        \Comment{Repeat each pattern $r$ times}
        \State $\ket{\psi} \leftarrow \text{Oracle()}$
        \Comment{Obtain copy of unknown state}
        \For{$i = 1$ to $n$}
            \Comment{Measure each qubit}
            \State Measure qubit $i$ in basis specified by $x_{b,i}$
            \If{$x_{b,i} = 0$}
                \State Measure in Z basis
            \ElsIf{$x_{b,i} = 1$}
                \State Measure in X basis
            \ElsIf{$x_{b,i} = 2$}
                \State Measure in Y basis
            \EndIf
            \State $m_{u,b,i} \leftarrow$ measurement outcome
            \Comment{Record result}
        \EndFor
    \EndFor
\EndFor
\end{algorithmic}

\vspace{2pt}
\hrule height 0.5pt

\end{minipage}
\end{figure}

\begin{figure}[h]
\centering
\begin{minipage}{0.95\columnwidth}

\hrule height 0.5pt
\vspace{4pt}

\textbf{Algorithm 2:Identifying Stabilizer Generators}

\vspace{4pt}
\hrule height 0.5pt
\vspace{2pt}

\begin{algorithmic}[1]
\Require  
$(\mathcal{X})_b$: Measurement pattern matrix
\newline
$(\mathcal{M})_{u,b}$: Measurement outcome matrix
\Ensure
$G$: Set of $n$ generators of the stabilizer state, including their respective signs

\State $G \gets \emptyset$ \Comment{Initialize generator set}
\For{$b = 1$ to $m$} \Comment{Iterate over measurement patterns}
    \State $\mathcal{S}_w \gets \left\{ S \subseteq \{1,\dots,n\} : |S| \le w \right\}$ \Comment{Generate subset of the power set with Hamming weights at most $w$}
    \For{each subset $s \in \mathcal{S}_w$}
        \State $P \gets \{0: 0, 1: 0\}$ 
        \Comment{Initialize parity counts}
        \For{$u = 1$ to $r$} \Comment{Iterate over repetitions}
            \State parity $\gets \bigoplus_{q \in s} (M)_{u,b,q}$ \Comment{Calculate parity}
            \State $P[\text{parity}] \gets P[\text{parity}] + 1$ \Comment{Update parity counts}
        \EndFor
        \If{$|P[0] - P[1]| > r/4$} \Comment{Check for significant difference}
            \State sign $\gets$ `+' if $P[0] > P[1]$ else `-' \Comment{Determine stabilizer sign}
            \State generator $\gets$ sign
            \For{$i = 1$ to $n$} \Comment{Construct stabilizer generator}
                \If{$i \in s$}
                    \State generator $\gets$ append$\big($generator , PauliOperator($(X)_{b,q}$)$\big)$
                \Else
                    \State generator $\gets$ append$\big($generator , `I'$\big)$
                \EndIf
            \EndFor
            \State $G \gets G \cup \{\text{generator}\}$ \Comment{Add generator to set}
        \EndIf
    \EndFor
\EndFor
\State \Return $G$ \Comment{Return the generator set}
\end{algorithmic}

\vspace{2pt}
\hrule height 0.5pt

\end{minipage}
\end{figure}

\subsection{Proof of Theorem \ref{th:Nbound}}\label{sec:proof1}
\begin{proof}
 Of all possible Pauli operators of weight at most $w$, the probability of learning $G_1$ is $1/A$, therefore, the probability of not learning $G_1$ is $1-1/A$.
After $N$ queries, the probability of not learning $G_1$ is $(1-1/A)^N$.
We need to find all $g$ stabilizer generators, so our algorithm fails if we do not find any one of $G_i, \; i=1,\dots, g$. The probability of this happening is 
\begin{equation}\label{eq:probfail}
 \Pr(\text{fail}) =   g \left( 1-\frac{1}{A}\right)^N.
\end{equation}
Using the inequality $(1-\frac{1}{A}) \leq e^{-1/A}$ for $A>1$, we get
\begin{equation}
\Pr(\text{fail}) \leq g e^{-N/A}.
\end{equation}
Now, let the target worst-case failure probability be $\delta$,
then let us impose the inequality
\begin{equation}
g e^{-N/A} \le \delta.
\end{equation}
Since the exponential function is a monotone function, we take the natural logarithm of both sides of the above inequality to get
\begin{align}
 \ln g - N/A \le \ln \delta
\end{align}
which is equivalent to 
\begin{align}
    - N/A  \le - \ln g + \ln \delta,
\end{align}
therefore
\begin{align}
    N  \ge A \ln g + A \ln(1/\delta).
\end{align}
For $\delta =1 / g$, we get
\begin{align}
    N  \ge 2  A \ln g .
\end{align}
\end{proof}

\subsection{Proof of Lemma \ref{lem:q}}\label{sec:proof2}
\begin{proof}
When performing a product measurement of a Pauli operator of weight $w$ (that is, a single-qubit measurement on each of the $w$ qubits in its support), if an even number of bit-flip errors occur, we will still obtain the correct measurement outcome. The probability of an incorrect measurement outcome is therefore the sum of the probabilities of any odd number of bit-flip errors occurring
\begin{equation}
\begin{split}
    q &= \sum_{\substack{ 1 \leq k\leq w \\ k \; \text{odd}}} \binom{w}{k}p^k(1-p)^{w-k}.
\end{split}
\end{equation}
By the method of generating functions, this is just 
\begin{align}
  q=   ( (1-p) + p)^w - ((1-p)-p)^w ) /2,
\end{align}
which when simplified gives the result.
\end{proof}

\subsection{Proof of Lemma \ref{lem:paccept}}\label{sec:proof3}
\begin{proof}
Let us first look at the case $Q\in\left\{ I,X,Y,Z \right\}^{\otimes n}$. Then the expected even parity count, when no parity errors occur, is $\mathbb{E}\left( P[0] \right) = r$ and the expected odd parity count is $\mathbb{E}\left( P[1] \right) = 0$. However, when we allow for parity errors to occur with probability $q$, then the expected parity counts become $\mathbb{E}\left( P[0] \right) = r(1-q)$ and $\mathbb{E}\left( P[1] \right) = rq$. We assume that $q$ is small and therefore $P[0]>P[1]$, therefore we can calculate the probability that $Q$ is rejected
\begin{equation}
\begin{split}
\Pr \left( |P[0]-P[1]|\leq \frac{r}{4} \right) &= \Pr\left( 2P[0]-r \leq \frac{r}{4} \right)\\
&=\Pr\left( P[0] \leq \frac{5r}{8} \right).
\end{split}
\end{equation}
We can then use the multiplicative Chernoff bound \cite{mitzenmacher2017probability}
\begin{equation}\label{eq:chernoff}
\Pr(X\leq (1-\delta)\mu)\leq \left( \frac{e^{-\delta}}{(1-\delta)^{1-\delta}} \right)^{\mu},
\end{equation}
where we choose $\mu = \mathbb{E}(P[0]) = r(1-q)$ and $\delta = 1-\tfrac{5}{8(1-q)}$ to get
\begin{equation}
\Pr\left( P[0]\leq \frac{5r}{8} \right) \leq c_q^{5r/8}.
\end{equation}
Similarly, if $Q\in-\left\{ I,X,Y,Z\right\}^{\otimes n}$ then $\mathbb{E}(P[1]) = r(1-q)$ and $\mathbb{E}(P[0]) = rq$. By the same reasoning we find that 
\begin{equation}
\Pr\left( P[1] \leq \frac{5r}{8} \right) \leq c_q^{5r/8}.
\end{equation}
Therefore if $Q\in \pm \left\{ I,X,Y,Z \right\}^{\otimes n}$ then the probability that our algorithm will correctly accept $Q$ as an element of the stabilizer group $G$ is 
\begin{equation}
\Pr(\text{accept}|Q\in G) \geq 1-2 c_q^{5r/8}.
\end{equation}
\end{proof}

\subsection{Proof of Lemma \ref{lem:preject}}\label{sec:proof4}
\begin{proof}
If we take measurements of $r$ copies of a Pauli operator $Q \notin G$, parity errors are equally likely to occur to both the even and odd parities, therefore the expectation values of the parity counts are $\mathbb{E}(P[0]) = \mathbb{E}(P[1]) = r/2$. There are two possible cases, either $P[0]>P[1]$ or $P[1]>P[0]$. Starting with the case $P[0]>P[1]$
\begin{equation}
\begin{split}
\Pr\left(|P[0]-P[1]|\leq \frac{r}{4} \right) &= \Pr\left( 2P[0] - r \right)\\
&= \Pr\left( P[0]\leq \frac{5r}{8} \right).
\end{split}
\end{equation}
Then, using the multiplicative Chernoff bound in eq. \eqref{eq:chernoff} with $\mu = \mathbb{E}(P[0]) = r/2$ and $\delta = 1/4$, we get
\begin{equation}
\Pr\left( P[0] \leq \frac{5r}{8}\right) \leq \left(\frac{64}{27 e}  \right)^{r/8}.
\end{equation}
For the case of $P[1]>P[0]$, similarly we get
\begin{equation}
\begin{split}
\Pr\left( |P[0]-P[1]|\leq \frac{r}{4} \right)&= \Pr\left( 2P[1]-r\leq \frac{r}{4} \right)\\
&= \Pr\left( P[1]\leq \frac{5r}{8} \right)\\
&\leq \left(\frac{64}{27 e}  \right)^{r/8}.
\end{split}
\end{equation}
Therefore, the probability of rejecting a Pauli operator, $Q$, that is not an element of the stabilizer group $G$ is
\begin{equation}
\Pr\left( \text{reject}| Q\notin G \right) \geq 1-2\left(\frac{64}{27 e}  \right)^{r/8}.
\end{equation}
\end{proof}

\subsection{Proof of Theorem \ref{th:pfailure}}\label{sec:proof6}
\begin{proof}
Let us take the coarse upper bound $A<w(3n)^w$ which comes from the fact that
\begin{align*}
\binom{n}{k} \leq n^k  \implies \binom{n}{k}3^k \leq (3n)^k \implies \sum_{k=1}^w \binom{n}{k}3^k \leq w(3n)^w.
\end{align*}
We will use \eqref{pfail} for the upper bound on the probability of failure. 
Noting that $64/(27e)<9/10$, the failure probability is bounded above by $n^{-2w}$ if the following inequality holds
\begin{equation}
 4w^2(3n)^{2w} \ln(n) (0.9)^{r/8} \le n^{-2w}.
\end{equation}
Taking the natural logarithm of both sides gives
\begin{equation}
\ln(4w^2)+2w\ln(3n)+\ln(\ln(n))+\frac{r}{8}\ln\left(\frac{9}{10}\right) \leq -2w\ln(n) , 
\end{equation}
therefore 
\begin{equation}
\frac{r}{8}\ln\left(\frac{9}{10}\right) \leq -2w\ln(3n^2) - \ln(4w^2)-\ln(\ln(n)).
\end{equation}
Then, noting that $\ln(9/10)<0$, we can find a bound on $r$
\begin{equation}
\begin{split}
r &\geq \frac{-8\left( 2w\ln(3n^2) + \ln(4w^2)+\ln(\ln(n)) \right)}{-\ln(10/9)}\\
&\geq 76\left(2w\ln(3n^2) + \ln(4w^2)+\ln(\ln(n)) \right),
\end{split}
\end{equation}
proving the lemma.
\end{proof}

\section{Conclusions}
In conclusion, we have presented an algorithm for learning the stabilizer generators that define any stabilizer code from
$N$ states belonging to the codespace. Our approach requires no prior knowledge of the code's structure beyond access to copies of states in the codespace, making it broadly applicable across different families of stabilizer codes. 
We defined the probability of a parity error occurring, which allowed us to obtain a bound on the algorithm's overall probability of success.
This method is particularly successful for qLDPC codes whose stabilizer generators are all of weight at most $w$. To learn the stabilizer generators with probability $1/g$, it suffices to take $N\geq A \ln(g)$, where $A$ is the number of Pauli operators of weight at most $w$. Thus, the learning problem remains tractable for codes with sparse stabilizer generators.
We have also shown that with query complexity polynomial in $n$, we may verify any stabilizer qLDPC code without prior knowledge of its generators.
An interesting open question is whether similar learning algorithms can be developed for non-stabilizer quantum codes.

\section{Acknowledgements}\label{eq:acknow}
Y.O. and H.L. acknowledge support from EPSRC (Grant No. EP/W028115/1). 
Y.O. also acknowledges support from the EPSRC funded QCI3 Hub under Grant No. EP/Z53318X/1.

\bibliography{ref}

\end{document}